\documentclass[letterpaper, DIV=11]{scrartcl}

\usepackage{amssymb, amsthm, mathtools}
\usepackage[hidelinks, citecolor=blue]{hyperref}
\hypersetup{colorlinks=true, urlcolor = blue, linkcolor=black}
\usepackage[numbers,sort&compress]{natbib}
\newtheorem{theorem}{Theorem}
\newtheorem{corollary}[theorem]{Corollary}
\newtheorem{lemma}[theorem]{Lemma}
\newtheorem{definition}[theorem]{Definition}
\newtheorem{proposition}[theorem]{Proposition}

\numberwithin{equation}{section}
\numberwithin{theorem}{section}

\newcommand{\rr}{{\mathbb{R}}}

\newcommand{\zz}{{\mathbb{Z}}}
\newcommand{\nn}{{\mathbb{N}}}
\newcommand{\cp}{{\mathbb{C}_+}}
\newcommand{\ee}{{\mathbb{E}\,}}
\newcommand{\eesub}[1]{{\mathbb{E}_{{#1}}}}
\newcommand{\conee}[2]{{\mathbb{E}\left[ \left. #1 \, \right\vert \,  #2 \right]}}
\newcommand{\pp}{{\mathbb{P}}}
\newcommand{\oh}{{\mathcal{O}}}
\newcommand{\im}{{\operatorname{Im}\,}}
\newcommand{\re}{{\operatorname{Re }\,}}
\newcommand{\parder}[2]{{\frac{\partial #1}{\partial #2}}}
\newcommand{\tr}{{\operatorname{Tr}\,}}
\newcommand{\beq}[1]{\begin{equation} \label{#1}}
\newcommand{\eeq}{\end{equation}}
\newcommand{\gtil}{{\tilde{G}}}

\begin{document}
\addtokomafont{author}{\raggedright}
\title{\raggedright Delocalization and continuous spectrum for ultrametric random operators}
\author{\hspace{-.075in}Per von Soosten and Simone Warzel}
\date{\vspace{-.2in}}
\maketitle
\minisec{Abstract} This paper studies the delocalized regime of an ultrametric random operator whose independent entries have variances decaying in a suitable hierarchical metric on $\mathbb{N}$. When the decay-rate of the off-diagonal variances is sufficiently slow, we prove that the spectral measures are uniformly $\theta$-H\"{o}lder continuous for all $\theta \in (0,1)$. In finite volumes, we prove that the corresponding ultrametric random matrices have completely extended eigenfunctions and that the local eigenvalue statistics converge in the Wigner-Dyson-Mehta universality class.
\bigskip

\section{Introduction}\label{sec:introduction}
Establishing regimes of eigenfunction delocalization, let alone determining the precise location of a localization-delocalization transition, in random matrix models with non-trivial spatial structure remains a major challenge. At one extreme end of this class of problems is the famous Anderson model, whose delocalization regime is only understood for the special case of tree graphs and similar structures~\cite{MR3055759,MR2257129,MR2274470,MR1302384, MR2864550,MR3510466, bauerschmidt2,bauerschmidt1}. Random band matrices~\cite{PhysRevLett.64.1851,PhysRevLett.67.2405} and power-law random band matrices (PRBM)~\cite{PhysRevE.54.3221} are more amenable to analysis, but, with the exception of a very recent result on Gaussian block-band matrices~\cite{MR3824956}, most mathematically rigorous methods for band matrices~\cite{MR2726110, MR3665217, MR3695802, MR3085669, MR2525652, hemarcozzi, bandmatrix1, bandmatrix2, bandmatrix3, doi:10.1093/imrn/rnx145} break down far from the critical point. In order to understand mechanisms for the occurrence of delocalized phases, it is therefore reasonable to study even simpler models like the ultrametric ensemble of Fyodorov, Ossipov and Rodriguez \cite{1742-5468-2009-12-L12001}. This structured random matrix ensemble has a variance profile that decays in a suitable hierarchical metric, mimicking the variance profile of the PRBM, whose off-diagonal entries decay in the Euclidean metric.

Since the work of Dyson~\cite{MR0436850}, hierarchical models like the ultrametric ensemble have figured prominently by paving the way for the mathematical analysis of more complicated models. Historically, hierarchical approximations have consistently reproduced qualitative features when inserted into the central models of statistical physics while remaining significantly more amenable to rigorous analysis (see for example~\cite{MR1552611, MR1143413, MR649813} and references therein). This principle notably fails for the hierarchical Anderson model~\cite{MR1063180}, which does not have a delocalized regime~\cite{MR1463464,MR2352276, proceedings, MR3649447} because of the localizing properties of the hierarchical Laplacian. Nevertheless, due to the random off-diagonal entries, the ultrametric ensemble is significantly less rigid than the hierarchical Anderson model and is expected to retain the core features of the PRBM localization transition~\cite{1742-5468-2009-12-L12001, PhysRevE.98.042116}. 

The ultrametric operator $H: \ell^2(\nn) \to \ell^2(\nn)$ is defined by the sum
\beq{eq:defHinfty} H = \sum_{r=0}^\infty \sqrt{t_r} \Phi_{\infty,r} , \qquad t_r = 2^{-(1+\epsilon)r},\eeq
where $\epsilon > -1$ is a parameter governing the decay rate of the off-diagonal entries. The hierarchical layers $\Phi_{\infty ,r}: \ell^2(\nn) \to \ell^2(\nn)$ have Gaussian entries that are independent up to the symmetry constraint with variances given by
\beq{eq:phidef} \ee \left| \langle \delta_y, \Phi_{\infty,r} \delta_x \rangle \right|^2 = \frac{1}{2^r} \begin{cases} 2 & \mbox{ if } d(x,y) = 0\\ 1 & \mbox{ if } 1 \le d(x,y) \le r\\ 0 & \mbox{ otherwise. }\end{cases}\eeq
The metric $d$ in~\eqref{eq:phidef} is the ultrametric
\[d(x,y) = \min  \left\{ r \geq 0 \, | \, \mbox{$x$ and $y$ lie in a common member of $\mathcal{P}_r$} \right\},\]
where $\{\mathcal{P}_r\}$ is the nested sequence of partitions defined by
\[\nn = \{1, \dots, 2^r\} \cup \{2^r + 1, \dots, 2\cdot 2^r \} \cup \dots\]
In~\eqref{eq:phidef} and throughout this paper, $\delta_x \in \ell^2(\nn)$ denotes the standard site basis element defined by
$ \delta_x(u) =1 $ when $ u = x $ and $\delta_x(u) = 0 $ otherwise. The natural counterpart of  $H$ in the finite volume
\[B_n = \{1,  \dots N_n \}, \qquad N_n = 2^n\]
is given by the ultrametric ensemble
\[H_n = \sum_{r=0}^n \sqrt{t_r} \Phi_{n, r}\quad 
\mbox{with} \quad \Phi_{n,r} = 1_{B_n} \Phi_{\infty, r} 1_{B_n}. \]

It is easy to check that the entries of $H$ decay according to
\[ \ee  \left| \langle \delta_y, H \delta_x \rangle \right|^2 = \oh \left(2^{-(2 + \epsilon)d(x,y)} \right) . \]
Since the hierarchical metric $d(x,y)$ grows only logarithmically in the volume, whereas the Euclidean metric grows linearly in the volume, this decay rate is the correct hierarchical analogue of the PRBM decay rate $ \ee \left| \langle \delta_y, H \delta_x \rangle \right|^2 = \oh \left(|x-y|^{-(2 + \epsilon)} \right) $. It is therefore reasonable to suppose that, like the PRBM, the ultrametric ensemble transitions between delocalization and localization as $\epsilon$ increases across the critical point $\epsilon = 0$. This critical point has a natural connection to the mixing properties of Dyson Brownian motion~\cite{MR0148397} (see~\eqref{eq:dbm} below) and is also in line with the theoretical physics predictions of~\cite{1742-5468-2009-12-L12001} and the numerical work~\cite{PhysRevE.98.042116}. 

The main results of this paper rigorously establish a delocalized phase when $\epsilon \in (-1, -1/2)$.  The precise formulations of these results require the average density of states $ \rho_n  $ of $ H_n $ defined by
\[\ee \frac{1}{N_n} \tr f(H_n) = \int \! f(\lambda) \rho_n(\lambda) \, d\lambda \]
for any continuous function $ f $. The existence of $ \rho_n \in L^1 \cap L^\infty $ is guaranteed by the Wegner estimate~\cite{MR639135}. 
Throughout the paper, we will fix a bulk spectral window $ W $, which is characterized by the following requirement.
\begin{definition} \label{def:bulk}
A closed interval $ W \subset \rr $ is a \emph{bulk set} if there is some neighborhood $\tilde{W}$ of $W$ such that
\beq{eq:dosassumption}\liminf_{n \to \infty} \inf_{E \in \tilde{W}} \rho_n(E) > 0.\eeq
\end{definition}
An explicit proof of the existence of such an interval would require a proof of some additional regularity of the density of states beyond the Wegner estimate, which we leave as an open problem. 

In the infinite volume, a principal probe for delocalization concerns the continuity of the spectrum of the random operator $H $. We note that $H$ is almost surely essentially self-adjoint on the functions of finite support, since the successive removal of off-diagonal matrix elements in $ B_r $ completely disconnects the system (a detailed proof is contained in Section~\ref{sec:proof}). Our first main result shows that the spectral measures of $H$ are almost surely $\theta$-H\"{o}lder continuous in the bulk for any exponent $\theta \in (0,1)$. It adds the ultrametric operator $ H $ in the regime $\epsilon \in (-1, -1/2)$ to the few existing examples of infinite-volume random operators for which existence of continuous spectrum is established.

\begin{theorem}\label{thm:measures} Suppose $\epsilon \in (-1, -1/2)$, let $x \in \nn$, and let $\mu_x$ denote the spectral measure of $\delta_x$ for $H$. Then, for any bulk set $W$ and any $ \theta \in (0,1) $,  there almost surely exists $C_x < \infty$ such that
\[\mu_x(E - \eta, E + \eta) \le C_x \eta^{1-\theta}\]
for all $E \in W$ and $\eta > 0$.
\end{theorem}
 
Theorem~\ref{thm:measures} has the usual dynamical consequence~\cite{MR1423040} that any observable $A$ on $ \ell^2(\nn) $ obeys the inequality
\[\int_0^T \! \left| \langle \delta_x, A_W(t) \delta_x \rangle \right| \, \frac{dt}{T} \le (CT^{-(1-\theta)})^{1/p} \|A\|_p,\]
where $A_W(t) = \left(1_W(H) e^{itH} \right) A  \left(1_W(H) e^{-itH} \right)$ and $\|A\|_p$ is the Schatten $p$-norm of $A$. In particular, choosing $A = 1_{B_R(x)}$ shows that energy-filtered dynamics satisfy
\[\int_0^T \! \sum_{y \in B_R(x)} \left| \langle \delta_y,  1_W(H) e^{-itH} \delta_x \rangle \right|^2 \, \frac{dt}{T} \le \left(CT^{-(1-\theta)} |B_{R}(x)|\right)^{1/p},\]
giving an explicit decay rate for the time-averaged quantum probability of measuring the position of a particle started at $x \in \nn$ in the hierarchical ball $B_R(x)$.

Our second main result is a refinement of Theorem~\ref{thm:measures} in finite-volumes, which shows that the normalized eigenfunctions of $ H_n $ in the bulk are maximally extended throughout the volume $B_n$. The formulation of this result uses the stochastic domination language of~\cite{MR3068390}, which we retain throughout this paper. Given two $N_r$-dependent random variables $X$ and $Y$, we say that $X$ is stochastically dominated by $Y$, written $X \prec Y$, if for every $\theta, p > 0$ we have
\begin{equation}\label{eq:notation}
\pp\left(X \geq N_r^\theta Y\right) \le N_r^{-p}
\end{equation}
for all sufficiently large $r$. 

\begin{theorem}\label{thm:eigenfunctions} If $\epsilon \in (-1, -1/2)$, then the $ \ell^2$-normalized eigenfunctions of $H_n$ with eigenvalues in any bulk set $ W $ satisfy
\[\sup_{E \in W \cap \sigma(H_n)} \|\psi_E\|_\infty  \prec N_n^{-1/2}.\]
\end{theorem}

This result stands in stark contrast to the behavior in the analogous regime of the closely related Rosenzweig-Porter model, where the eigenfunctions are extended but not uniformly across the volume~\cite{nonergodic, benignibourgade}. We thus confirm recent numerical work \cite{PhysRevE.98.042116} suggesting that the ultrametric ensemble does not possess an intermediate phase in the regime $ \epsilon \in (-1,0) $.

The proofs of these theorems take a dynamical perspective, using the fictitious dynamics obtained by representing the Gaussian perturbations $\Phi_{n,n}$ in terms of Brownian motion in the real symmetric matrices. It is not hard to see that the ultrametric ensemble can be constructed in a recursive manner by initializing a $1 \times 1$ random $ \mathcal{N}(0,2) $ entry $H_0$ and setting
\beq{eq:dbm} H_n = H_{n-1} \oplus H_{n-1}^\prime + \Phi_{N_n}(t_n).\eeq
Here, $H_{n-1}^\prime$ is an independent copy of $H_{n-1}$ and
\beq{eq:goe} \langle \delta_y, \Phi_N(t) \delta_x \rangle = \sqrt{\frac{1 + \delta_{xy}}{N}} \, B_{xy}(t),\eeq
with a symmetric array of standard Brownian motions $\{B_{xy}\}$ indexed by $ x ,y \in \{ 1,\dots , N\} $. From this point of view, the conjectured critical point $\epsilon = 0$ is natural since it corresponds to the local equilibration time $N_n^{-1}$ of the Dyson Brownian motion~\cite{MR2810797} governing the evolution of the spectrum under~\eqref{eq:dbm}.

Our analysis is based on strong probabilistic estimates for the local resolvent entries
\[G_n(x;z) = \langle \delta_x, (H_n-z)^{-1} \delta_x \rangle\]
on almost microscopic scales $\im z \approx N_n^{-1 + \alpha}$ with $ \alpha > 0 $ small.  Given some control on the same scale of the Stieltjes transform of the empirical eigenvalue measure,
\[S_n(z) = \frac{1}{N_n} \tr (H_n-z)^{-1},\]
our analysis shows that the effect of the Gaussian perturbations $\Phi_{N_n}(t_n)$ is to slowly lift the complex energy parameter in the complex plane:
\[G_n(x,z) \approx G_{n-1}(x, w), \quad \im w \approx \im z + \mathcal{O}(t_n).\]
Iterating this observation and applying trivial bounds to the resolvent show that $G_n(x;z)$ grows very slowly in $n$ even on almost microscopic scales. Such a bound immediately yields the delocalization of the eigenvectors in Theorem~\ref{thm:eigenfunctions}, whereas the continuity of the spectral measures in Theorem~\ref{thm:measures} now follows from little more than weak convergence.

The aforementioned local control on the spectral measures has featured prominently in the random matrix theory literature, where it is the first step of the ``three-step-strategy'' of Erd\"{o}s, Schlein, and Yau~\cite{MR2810797} for proving Wigner-Dyson-Mehta universality of the local statistics. Due to the Gaussian nature of our problem, only one additional step is needed for universality, which is contained in the work of Landon, Sosoe, and Yau~\cite{landonsosoeyauarxiv}. The objects of interest are the local versions of the $k$-point correlation functions, defined as the $k$-th marginals of the symmetrized eigenvalue density $\rho_{H_n}$:
\beq{eq:rhokdef}\rho^{(k)}_{H_n}(\lambda_1, .. \lambda_k) =  \int_{\rr^{2^n - k}} \! \rho_{H_n} (\lambda_1, \dots, \lambda_{2^n}) \, d\lambda_{k+1} \dots \, d\lambda_{2^n}.\eeq
The precise formulation of the theorem considers the difference
\begin{align}\label{eq:psidef} \Psi_{n,E}^{(k)}(\alpha_1, \dots ,\alpha_k) &= \rho^{(k)}_{H_n}\left( E + \frac{\alpha_1}{N_n \, \rho_{n, fc}(E)}, \dots, E + \frac{\alpha_k}{N_n \,\rho_{n, fc}(E)} \right)\nonumber \\
 &-  \rho^{(k)}_{GOE}\left(\frac{\alpha_1}{N_n \,\rho_{sc}(0)}, \dots,  \frac{\alpha_k}{N_n \, \rho_{sc}(0)} \right).
\end{align}
Here, $\rho^{(k)}_{GOE}$ denotes the $k$-point correlation function of the $N_n \times N_n$ Gaussian Orthogonal Ensemble and $\rho_{sc}$ is the density of the semicircle law. The function $\rho_{n, fc}(E)$ is defined as the density of the measure whose Stieltjes transform $M(z)$ solves
\beq{eq:Mdef} M(z) =\int\! \frac{\rho_{n-1}(\lambda)}{\lambda - z - t_n M(z)} \, d\lambda.\eeq

\begin{theorem}\label{thm:localstats} Suppose $\epsilon \in (-1, -1/2)$, $E \in W$ and $k \geq 1$. Then
\[\lim_{n \to \infty} \int_{\rr^k} \! O(\alpha) \Psi_{n,E}^{(k)}(\alpha) \, d\alpha = 0\]
for every $O \in C_c^\infty(\rr^k)$.
\end{theorem}

We conclude this introduction by contrasting our results with those of~\cite{resflow}, which proved that
\begin{itemize}
\item in the regime $\epsilon \in (0, \infty) $, the eigenvectors are localized and the local statistics converge to a Poisson point process.
\item there is a mean-field regime $ \epsilon \in (-\infty, -1)$, for which the last term $ \sqrt{t_n} \,  \Phi_{n,n} $ in $H_n$ already forces delocalization with an estimate that degrades as $ \epsilon \to -1 $ (see also~\cite{MR3068390}) and Wigner-Dyson-Mehta universality of the local statistics (see also~\cite{landonsosoeyauarxiv}).
\end{itemize}
 
Since the sums defining $H_n$ do not remain bounded if $\epsilon \in (-\infty, -1)$, the results of~\cite{resflow} for this regime rescaled the Hamiltonian by a factor 
\[Z_n^2 =  \sum_{y \in B_n} \ee \Big| \langle \delta_y, \sum_{r =0}^n \sqrt{t_r} \, \Phi_{n,r}   \delta_x\rangle \Big|^2 = \oh\left(2^{-(1+\epsilon)n}\right).\]
that keeps the spectrum of $ H_n $ on order one. The methods of the present paper can easily be adapted to prove that Theorem~\ref{thm:eigenfunctions} also holds in the mean-field regime, yielding significantly improved bounds that do not blow up as $\epsilon \to -1$. We also note that, in the mean-field regime, any compact $ W \subset (-2,2) $ is a bulk set due to the validity of the semicircle law $ \sqrt{(4-E^2)_+}/(2\pi) $~\cite{MR3068390,resflow}.

Our results therefore leave open the regime $\epsilon \in (-1/2, 0)$, for which we expect all of the main results in this paper to remain valid. Indeed, the assumption $\epsilon < -1/2$ seems of a technical nature. The difficulty is that we cannot prove the necessary local bounds on the spectrum contained in Theorem~\ref{thm:locallaw} for $\epsilon \in (-1/2, 0)$. If we assumed the validity of this theorem for all parameters $\epsilon < 0$, our delocalization results would extend to the entirety of the regime as well. Nevertheless, Theorem~\ref{thm:locallaw} seems to be one of very few results of its kind that are valid when the model does not have a large spread in the sense of~\cite{MR3068390}.

The paper is organized as follows. Section~\ref{sec:schrodinger} proves some local bounds for a general random Schr\"{o}dinger operator and Section~\ref{sec:dbm} expands some previous results on Dyson Brownian motion~\cite{nonergodic} to better suit the current setting. These results figure in the local bounds for the ultrametric ensemble in Section~\ref{sec:locallaw} and the estimates for the Green functions in  Section~\ref{sec:localresbounds}. Finally, Section~\ref{sec:proof} translates these estimates into our main results.

\section{Local Bounds for Random Schr\"{o}dinger Operators}\label{sec:schrodinger}
In our proof of the local law for the ultrametric ensemble, we will require some concentration of measure results for the Stieltjes transform of the empirical eigenvalue measure of a generic random Schr\"{o}dinger operator. Since these bounds may be of some independent interest, this short section presents them in a more general setting. We consider
\[H = A + V\]
where $A$ is a fixed hermitian $N \times N$ matrix and
\[V = \operatorname{diag}(V_1, \dots, V_N)\]
are  independent random variables. Our goal is to understand the behavior of
\[S(z) = \frac{1}{N} \tr (H-z)^{-1}\]
on small scales $\im z = o(1)$ as $ N \to \infty $. The content of the first  theorem is that $ S(z)$ does not fluctuate on scales larger than $\im z \gg N^{-1/2}$. The statement uses $N$-independent constants $C, c \in (0, \infty)$, whose exact value may change from line to line. We will keep this convention for the rest of this paper.
\begin{theorem}\label{thm:schroedingerlaws} For any $ \mu > 0 $ and $z = E + i\eta$, 
\begin{equation}
\pp\left(\left| S(z) - \ee S(z) \right| > \mu \right) \le C\exp\left(-cN(\mu \eta)^2\right).
\end{equation}
\end{theorem}
\begin{proof}
We consider $ S(z) $ as a function of $ N $ independent random variables $ V = (V_1, \dots, V_N)$. If $V$ and $\tilde{V}$ differ in only one variable, then $H = A+V$ and $\tilde{H}  = A+\tilde{V}$ differ by a rank-one perturbation. Therefore, the eigenvalues of $H$ and $\tilde{H}$ interlace and
\begin{align}\label{eq:cgk} |S(z) - \tilde{S}(z)| &\le  \frac{1}{N} \int \!  \left| \lambda - z \right|^{-2} \left| \tr 1_{(-\infty, \lambda)}(H) - \tr 1_{(-\infty, \lambda)}(\tilde{H}) \right| \, d\lambda \le \frac{\pi}{N\eta}.
\end{align}
The claim now follows immediately from McDiarmid's inequality, which states that
\[\pp\left(|S(z) - \ee S(z) | > \lambda \sigma \right) \le C \, \exp\left(-c\lambda^2\right)\]
with $\sigma^2 = \pi^2 \sum_{x=1}^N (N\eta)^{-2}  = \pi^2 / \left(N \eta^2\right)$.
\end{proof}

The second theorem gives upper bounds on $ S(z) $ on the optimal scale $\im z \gg N^{-1}$. Essentially, it states that the Wegner estimate for random Schr\"odinger operators has exponential tails even down to mesoscopic scales in $ \im z $.
\begin{theorem}\label{thm:schroedingerlaws2} Suppose each $V_x$ is drawn from a common bounded density $\rho \in L^1 \cap L^\infty$. Then
\begin{equation}\label{eq:upperbd} 
\pp\left(\im S(z) >  c + \mu \right) \le \exp\left(-\mu  N\eta \right)
\end{equation}
and
\begin{equation}\label{eq:upperbd2} 
\pp\left( |S(z) |  >  c \left[ 1+  \log(1+ \eta^{-1}) \right] \right) \le \frac{C}{\eta^2} \exp\left(-  N\eta \right)
\end{equation}
for all $ \mu > 0 $ and $z = E + i\eta$.
\end{theorem}
\begin{proof}
For the proof of the first statement, we consider
\[g(t) = \ee \exp(t \, \im S(z)).\]
Letting $\eesub{x}$ denote the expectation over a single random variable $V_x$, we estimate
\begin{align*} g^\prime(t) &= \ee\left[ \exp(t \, \im S(z)) \frac{1}{N}\sum_x \im \langle \delta_x, (H-z)^{-1} \delta_x \rangle \right]\\
&\le \ee\left[ \exp\left(t \, \im S(z) + t\pi(N\eta)^{-1}\right) \frac{1}{N}\sum_x  \eesub{x} \, {\im \langle \delta_x, (H-z)^{-1} \delta_x \rangle}\right]\\
&\le \pi \|\rho\|_\infty \exp\left(t\pi(N\eta)^{-1} \right) g(t).
\end{align*}
using the fluctuation bound~\eqref{eq:cgk} and the spectral averaging principle (see~\cite{MR3364516}). By Gr\"{o}nwall's inequality we conclude that
\[g(t) \le \exp\left(t e^{t\pi(N\eta)^{-1}} \pi \|\rho\|_\infty \right),\]
so setting $t = N\eta $ and using an exponential Chebyshev estimate proves the first part of the theorem.

The second part follows from the representation
\[
\re S(z) =  \int \im S(E+t + i \eta/2) \, q_{\eta/2}(t) dt , \quad q_{\eta}(t) := \frac{1}{\pi} \frac{t}{t^2+\eta^2} , 
\]
which, after splitting the domain of integration, implies the bound
\begin{align*}
 |\re S(z) |    &\leq 1+ \sup_{t \in [-1,1]} \im S(E+t + i \eta/2) \int_{-1}^1 \left| q_{\eta/2}(s) \right| ds\\
 &=  1+  \frac{\log \left(1+ 4\eta^{-2} \right)}{\pi}  \sup_{t \in [-1,1]} \im S(E+t + i \eta/2).
\end{align*}
Since $ S(E+t + i \eta/2)  $ is $ 4\eta^{-2} $-Lipschitz continuous in $ t $, the claim now follows from~\eqref{eq:upperbd} and the union bound.
\end{proof}

We note here that when $A$ is the restriction of the Laplacian to some finite lattice in $\zz^d$, it is known (see \cite{1064-5616-206-1-93} and references therein) that $S(z)$ obeys the central limit theorem on macroscopic scales in the sense that  $\sqrt{N} \left( S(z) - \mathbb{E}[S(z)]\right)$ converges to a Gaussian random variable for fixed $ z \in\cp $.

\section{Results on Dyson Brownian Motion} \label{sec:dbm}
In this section we expand some results from~\cite{nonergodic}, which will be of use in controlling the Dyson Brownian motion~\eqref{eq:goe} in the iterative construction~\eqref{eq:dbm} of the ultrametric ensemble. In~\cite{nonergodic}, we considered
\[H(t) = V + \Phi_N(t), \]
where $\Phi_N(t) $ is a $ N \times N $ Dyson Brownian motion as in~\eqref{eq:goe} and $V =  \operatorname{diag}(V_1,\dots , V_N) $ is a deterministic initial condition that can be taken diagonal without loss of generality. We then studied the evolution of
\[S_t(z) = \frac{1}{N} \tr \left(H(t) - z\right)^{-1} , \qquad G_t(x,z) = \langle \delta_x , \left(H(t) - z\right)^{-1} \delta_x \rangle\]
up to a time $T = N^{-1 + \epsilon}$ for spectral parameters chosen from
\[\Omega = (W_1, W_2) + i(\eta, 10), \qquad \eta = N^{-1 + \alpha}\]
with $ \alpha \in (0,1) $ arbitrarily small. This was accomplished by tracking the flow along the random characteristic curves $\gamma(t, z)$ solving the initial-value problem
\begin{equation}\label{eq:ODE}
\dot{\gamma}(t, z) = -S_t(\gamma(t,z)), \qquad \gamma(0, z) = z.
\end{equation}
The next proposition summarizes two of the three main technical results of~\cite{nonergodic}. For its statement, we introduce the stopping time
\[\tau_{z} = \inf\{t: \im \gamma(t,z) = \eta/2\} \]
and the stopped characteristic
\[\xi_t(z) = \gamma(t \wedge \tau_z, z).\]
The first point below is the content of Theorem 2.1, while the second may be found in the proof of Theorem 2.3 of~\cite{nonergodic}. 

\begin{proposition}[cf.~\cite{nonergodic}]\label{prop:char} Let $z \in \Omega$ and let $z_1, z_2 \in \cp$.
\begin{enumerate}
\item
The change in $S_t$ along the characteristic curve satisfies
\beq{eq:mgboundformula} \pp\left(\sup_{t \le \tau_z} |S_t(\xi_t(z)) - S_0(z)| > \frac{4}{\sqrt{N \eta}} \right) \le C \exp\left(-c N \eta \right).
\eeq
\item The flow $\gamma(t,\cdot)$ obeys the Lipschitz-estimate
\begin{equation}\label{eq:Gronwall}
\left| \gamma(t,z_1) - \gamma(t,z_2) \right| \leq \sqrt{\frac{\im z_1 \, \im z_2 } { \im \gamma(t,z_1)\,   \im \gamma(t, z_2)}} \, | z_1-z_2 |
\end{equation}
as long as $ \im \gamma(t, z_1) , \im \gamma(t, z_2)> 0$.
\end{enumerate}
\end{proposition}

The two results may be combined to control $S_t(\xi_t(z))$ for a continuum of points simultaneously. For this purpose, let
\beq{eq:aevent} \mathcal{A} =  \left\{ \sup_{z \in \Omega } \sup_{t \le \tau_z} |S_t(\xi_t(z)) - S_0(z)| \leq \frac{C}{\sqrt{N \eta}} \right\} . \eeq
\begin{lemma} \label{thm:continuum} The event $\mathcal{A}$ satisfies
\begin{equation}\label{eq:uniformchar}
\pp( \mathcal{A}) \geq 1 - C\exp\left(-c N \eta \right).
\end{equation}
\end{lemma}
\begin{proof} Let $\Lambda \subset \Omega$ be a finite grid such that $|\Lambda| \le C N^8 $ and $\operatorname{dist}(z, \Lambda) \le C \eta^3 /\sqrt{N\eta} $ for all $z \in \Omega$. Let
\[\hat{\tau}_z = \inf \{t > 0: \im \gamma(t, z) = \eta/4\}.\]
It is clear from the proof of Proposition~\ref{prop:char} that also
\[\pp\left(\sup_{t \le \hat{\tau}_z} |S_t(\xi_t(z)) - S_0(z)| > \frac{C}{\sqrt{N \eta}} \right) \le C \exp\left(-c N \eta \right)\]
with slightly altered constants. By the union bound, the event
\[\tilde{\mathcal{A} }=  \left\{ \sup_{z \in \Lambda } \sup_{t \le \hat{\tau}_z} |S_t(\xi_t(z)) - S_0(z)| \leq \frac{C}{\sqrt{N \eta}} \right\} \]
satisfies $ \pp( \tilde{\mathcal{A}}) \geq 1 - C N^8 \exp\left(-c N \eta \right) $ so it suffices to show that $\mathcal{A} \subset \mathcal{\tilde{A}}$. If $z \in \Omega$, let $z_0 \in \Lambda$ be such that $|z - z_0| \le C \eta^3 /\sqrt{N\eta} $. By the triangle inequality,
\[|S_t(\xi_t(z)) - S_0(z)| \leq |S_t(\xi_t(z_0)) - S_0(z_0)| + |S_0(z_0)- S_0(z)|  + |S_t(\xi_t(z)) - S_t(\xi_t(z_0))|.\]
The first term is bounded by $C/\sqrt{N\eta}$ by the definition of $\tilde{\mathcal{A}}$. The second is bounded by $ C \eta^{-2} |z-z_0| $ because of the Lipschitz-continuity of $S_0$. Finally, the third term is bounded by $ C \eta^{-3} |z - z_0| $ using the Lipschitz-continuity of $ S_t $ and the estimate~\eqref{eq:Gronwall}. This concludes the proof.
\end{proof}

Our next goal is to show that one can use the previous results to propagate a lower bound at time $t = 0$ in a spectral domain 
\[\tilde{\Omega} = (\tilde{W_1}, \tilde{W_2}) + i(\tilde{\eta_1}, \tilde{\eta_2}), \qquad \tilde{\eta}_1 = N^{-1 + \tilde{\alpha}}\]
with a coarser spectral parameter $  \tilde{\eta_1}  > \eta $ to a lower bound at time $t = T$ into the domain~$\Omega$ with the finer spectral parameter $\eta $. This result is valid if there are some $ K_l , K_u  \in (0,\infty) $ satisfying the initial bounds
\beq{eq:s0bound}
K_l \leq \inf_{z\in \tilde{\Omega}} \im S_0(z) \quad\mbox{and}\quad \sup_{z\in  \Omega}  |S_0(z) | \leq K_u \log N 
\eeq
and the compatibility conditions
\beq{eq:Tlower}
T K_l  > 2 \tilde{\eta_1}\quad\mbox{and }\quad 2 K_u  T \log N  \leq \min\left\{ W_1 -\tilde{W_1}  ,  \tilde{W_2} - W_2 , \tilde{\eta_2}  - 10 \right\} .
\eeq
\begin{theorem}\label{thm:bndprop}
Suppose that the event $ \mathcal{A} $ occurs and that~\eqref{eq:s0bound} and \eqref{eq:Tlower} hold. Then the following statements are valid for all $z \in \Omega$ when $N$ is sufficiently large.
\begin{enumerate}
\item There exists $w \in \tilde{\Omega}$ such that $\xi_T(w) = z$ and $\tau_w \geq T$.
\item $ \im S_{T } (z) \geq  K_l /2$.
\end{enumerate} 
\end{theorem}
\begin{proof}
Let $ z \in \Omega $ with $ \im  \gamma(T,z) > \eta/2 $. If $ \mathcal{A} $ occurs, integrating the characteristic~\eqref{eq:ODE} shows that
 \[
 \left| \gamma(T,z) - z - T S_0(z) \right| \leq \frac{C T}{\sqrt{N\eta} }.
 \]
In particular, the trajectory $\gamma(t, z)$ has to obey the bound
 \beq{eq:gammabdupper} \left| \gamma(T,z) - z \right| \leq 2 K_u T \log N
 \eeq
when $ N $ is sufficiently large. If, in addition, $z \in \tilde{\Omega}$, then also
 \beq{eq:gammabdlower} \im \gamma(T,z) \le \im z -  T K_l /2
 \eeq
by the first inequality in~\eqref{eq:s0bound}. Let $\lambda$ denote the time-reversed flow corresponding to~\eqref{eq:ODE} so that
\begin{equation}\label{eq:ODErev}
 \dot\lambda(t,w) = S_{T-t}(\lambda(t,w) ), \quad \lambda(0, w) = w .
\end{equation}
 for $ t \leq T $. By construction, we have $ \gamma(T,\lambda(T,w)) = w $ so the first assertion will follow from $ \lambda(T,\Omega) \subset  \tilde{\Omega} $. To prove this, note that $ \lambda(T,\Omega) $ is  simply connected and open in $\cp $, since the flow~\eqref{eq:ODErev} is a homeomorphism. It follows from~\eqref{eq:gammabdupper} and the second inequality in~\eqref{eq:Tlower} that the intersection of the boundary $ \partial \Omega $ with $  \tilde{\Omega}  $ is mapped to $  \tilde{\Omega}  $ under this flow and that
 \[\re \lambda(T, \Omega) \subset  (\tilde{W_1}, \tilde{W_2})\]
In order to show that also the remainder of the boundary $ \partial \Omega $ is mapped to  $  \tilde{\Omega}  $, suppose  the contrary. Then there would be a point $z_0 \in \lambda(T,\Omega) $ with $\im z_0 < \tilde{\eta_1}$. Since $ \lambda(T,\Omega) $ is simply connected, there must also be a point $z \in \lambda(T,\Omega) \cap \tilde{\Omega}$ with $\im z = \tilde{\eta}_1 + \varepsilon$ for an arbitrarily small $\varepsilon > 0$. By~\eqref{eq:gammabdlower}, such a point would have $\im \gamma(T, z) \le \varepsilon$.

The second assertion follows from the first together with Lemma~\ref{thm:continuum}. Indeed, we have
\[\inf_{z\in \Omega} \im S_T(z) \geq \inf_{z \in \tilde{\Omega}} \im S_{T \wedge \tau_z} (\xi_T(z)) \geq \inf_{z \in \tilde{\Omega}} \im S_0(z) - \frac{C}{\sqrt{N\eta}},\]
where we used that $\tau_z \geq T$ for every $z \in \lambda(T, \Omega)$.
\end{proof}

In addition to the resolvent trace, we will need to consider the evolution of the diagonal Green function along a characteristic curve. The relevant technical results for this are contained in the proof of Theorem 3.1 of~\cite{nonergodic}, which employs the additional stopping time
\[
\tau(z) :=  \tau_z   \wedge \inf \left\{t > 0:  \int_0^{t \wedge  \tau_z } \! \frac{ds}{(\im \xi_s(z))^2}  \geq \frac{5}{K_l  \eta} \right\} . 
\]
\begin{proposition}[cf.~\cite{nonergodic}] \label{thm:gfuncmoment} Let $q \geq 2$. If $N$ is sufficiently large, then
\beq{eq:charbound}
\ee{\left|\im G_{t \wedge \tau(z)}(x, \xi_{t \wedge \tau(z)}(z)) \right|^q} \le \left(1 - Cq \sqrt{\frac{10}{K_l N \eta }}\right)^{-q} \left| \im G_0(x;z) \right|^q 
\eeq
for all  $z \in \Omega$.
\end{proposition}
In~\cite{nonergodic}, the last bound was presented in the form
\[\ee{\left|\im G_{t \wedge \tau(z)}(x, \xi_{t \wedge \tau(z)}(z)) \right|^q} \le \left(1 - C_q \sqrt{\frac{10}{K_l N \eta  }}\right)^{-q} \left| \im G_0(x;z) \right|^q,\]
where the constant  $C_q$ is the constant for the Burkholder-Davis-Gundy inequality in the form
\[ \left\| \sup_{s \le t} |M_s|\right\|_q \le C_q \left\| \left[ M\right]_t^{1/2} \right\|_q.\] 
In the present paper, it will be important that this constant does not grow faster than $C_q \le Cq$ for $q \geq 2$  (see~\cite{MR1725357}), which we have already included in~\eqref{eq:charbound}. Finally, we note that Corollary~2.2 in~\cite{nonergodic} shows that the distinction between the stopping times $\tau_z$ and $\tau(z)$ is artificial on the event $ \mathcal{A} $. 
\begin{corollary}[cf.~\cite{nonergodic}]\label{thm:stoppingtimes}
In the setting of Theorem~\ref{thm:bndprop} we have $  \tau_z = \tau(z) $ for any $ z \in \Omega $.
\end{corollary}

\section{Local Bounds for the Ultrametric Ensemble}\label{sec:locallaw}
In this section we prove some local bounds for  the Stieltjes transform
\[S_r(z) = \frac{1}{N_r} \tr (H_r - z)^{-1}\]
of the empirical eigenvalue measure of the ultrametric ensemble valid uniformly in a spectral domain of the form
\[\Omega_r = W_0 + i(\eta_r, 10).\] 
The admissible set of energies $W_0 \subset \rr$ is an open interval satisfying
\[W \subset W_0  \subset \tilde{W}, \qquad \operatorname{dist}(W_0, \partial \tilde{W}) > 0 \]
where $W$ is the bulk set and $\tilde{W}$ is the neighborhood of $W$ from Definition~\ref{def:bulk}. The spectral scale is given by
\[\eta_r = N_r^{-1 + \alpha}\]
with some $\alpha > 0$ to be fixed later in Section~\ref{sec:proof} for the proofs of our main results.

\begin{theorem}\label{thm:locallaw} 
There are constants $K_l, K_u \in (0, \infty)$ such that for any $p > 0$ and large enough $r $ we have
\[\pp\left(\im S_r(\Omega_r) \subset [K_l, K_u]  \; \mbox{and} \;  |S_r(\Omega_r) | \subset(K_l, K_u  \log N_r) \right) \geq 1 - N_r^{-p}.\]
\end{theorem}

The upper bounds in the theorem can be established directly for all $z \in \Omega_r$ by regularizing the spectrum with only the diagonal randomness.

\begin{proof}[Proof of the upper bounds]
We will split
\begin{equation}\label{eq:split}
H_{r} = A + V
\end{equation}
into a sum of its diagonal entries $V$ and off-diagonal entries~$A$. Conditioning on the values of~$ A $, 
the upper bounds on $ \im S_r $ and $ |S_r | $  are immediate consequences of Theorem~\ref{thm:schroedingerlaws2}. 
\end{proof}

The proof of the lower bounds will proceed in two steps. The first step establishes some concentration of measure for $S_r(z)$ when the spectral parameter lies in a coarse spectral domain
\[\tilde{\Omega}_r = \tilde{W} + i(\tilde{\eta_r}, 10 + 2K_ut_r\log(N_r)),\]
where $ \tilde{\eta_r} > \eta_r $ and $\tilde{W}$  is the neighborhood of the bulk set $W$. Since $W$ is a bulk set,
\begin{equation}\label{eq:dospos} 
\inf_{z \in \tilde{\Omega}_r}  \ee\left[ \im S_r(z) \right] \geq  K_l
\end{equation}
for an appropriately chosen $K_l > 0$, which then implies a lower bound for $\im S_r$ uniformly in $\tilde{\Omega}_r$ with high probability. In the second step, this lower bound is propagated to the finer spectral domain $\Omega_r$ by applying the results of Section \ref{sec:dbm} to the Dyson Brownian motion in a single step of the iterative construction~\eqref{eq:dbm}.

Since we are considering the regime $\epsilon < -1/2$, it is possible to choose the coarse spectral scale
\[\tilde{\eta}_r = N_r^{-1 + \tilde{\alpha}}\]
such that
\beq{alphatilde} 1/2 < \tilde{\alpha} < -\epsilon.
\eeq
The lower bound $1/2 < \tilde{\alpha}$ enables the first step via the following lemma, whose proof is a combination of Theorem~\ref{thm:schroedingerlaws} and the strategy in~\cite{MR1781846}.

\begin{lemma}\label{thm:coarselocallaw} We have
\begin{equation} \label{eq:locall} 
 \sup_{z \in  \tilde{\Omega}_r} |S_r(z) - \mathbb{E}\left[S_r(z)\right]| \prec N_r^{1/2 - \tilde{\alpha}}.
\end{equation}
\end{lemma}
\begin{proof}
We again start by splitting $H_r$ into diagonal entries $V$ and off-diagonal entries~$A$ as in ~\eqref{eq:split}. For $z \in \tilde{\Omega}_r$, Theorem~\ref{thm:schroedingerlaws} shows that the conditional expectation
\[\eesub{V} S_r(z) = \conee{S_r(z)}{A}\]
satisfies
\beq{eq:shatbound} \big| S_r(z) - \eesub{V} S_r(z)\big| \prec  N_r^{1/2 - \tilde{\alpha}}.
\eeq
To show that $\eesub{V} S_r(z)$ is close to $\ee S_r(z)$, we note that the matrix $A$ can be written as an entrywise product $A_{xy} =: \Sigma_{xy}  \tilde{A}_{xy} $ of a symmetric matrix $\tilde{A}$ with independent $\mathcal{N}(0,1)$ entries and a profile $\Sigma$ satisfying
\[ \max_x \sum_y  \left|\Sigma_{xy}\right|^2 \le C\]
uniformly in $r \geq 0$. The partial derivatives of $\eesub{V}S_r(z)$ with respect to off-diagonal $\tilde{A}_{xy}$ may easily be calculated as
\[\parder{\eesub{V} S_r(z)}{\tilde{A}_{xy}} = - \frac{\Sigma_{xy}}{N_r} \, \eesub{V} \parder{}{z} \left[ \langle \delta_y, (H_{r}-z)^{-1} \delta_x \rangle +  \langle \delta_x, (H_{r}-z)^{-1} \delta_y \rangle \right].\]
Hence,
\begin{align*} \left| \parder{\eesub{V} S_r(z)}{\tilde{A}_{xy}} \right|^2 &\le \left(\frac{2 \Sigma_{xy}}{N_r} \, \eesub{V} \langle \delta_y, |H_{r}-z|^{-2} \delta_x \rangle \right)^2 \\
&\le \left(\frac{2 \Sigma_{xy}}{N_r}\right)^2  \eesub{V} \langle \delta_x, |H_{r}-z|^{-2} \delta_x \rangle \, \eesub{V} \langle \delta_y, |H_{r}-z|^{-2} \delta_y \rangle \\
&\le \left(\frac{2 \Sigma_{xy}}{N_r \, \im z}\right)^2  \eesub{V} \im G_r(x, x; z) \, \eesub{V}  \im G_r(y, y; z)  \le C \left(\frac{ \Sigma_{xy}}{N_r \, \im z}\right)^2,
\end{align*}
where we used the Cauchy-Schwarz inequality with respect to the joint probability and spectral measures and then applied the spectral averaging principle. We conclude that
\[ \Bigg(\sum_{x<y}  \Big| \parder{\eesub{V} S_r(z)}{\tilde{A}_{xy}} \Big|^2 \Bigg)^{1/2} \le \frac{C}{\sqrt{N_r (\im z)^2}}\]
so for $z \in \tilde{\Omega}_r$ we obtain
\[\pp\left( \left|\eesub{V} S_r(z) - \ee S_r(z)\right| > \frac{N_r^\mu}{\sqrt{N_r (\im z)^2} } \right) \le  C \exp\left(- c N_r^{2\mu} \right)\]
by the concentration inequality for Lipschitz-continuous functions of Gaussian random variables (see~\cite{MR2906465}). So we have proved that
\[ \big|\eesub{V} S_r(z) - \ee S_r(z)\big| \prec N_r^{1/2 - \tilde{\alpha}}.\]
and combining this with~\eqref{eq:shatbound} yields
\[\left| S_r(z) - \ee S_r(z)\right| \prec N_r^{1/2 - \tilde{\alpha}} \]
for every $ z \in \tilde{\Omega}_r $. Since $S_r$ is $\tilde{\eta}_r^{-2}$-Lipschitz continuous in $\tilde{\Omega}_r$, choosing a sufficiently fine grid of $z \in \tilde{\Omega}_r$ and applying the union bound proves~\eqref{eq:locall}.
\end{proof}

Turning to the second step, we recall~\eqref{eq:dbm} and write
\[H_{r} = H_{r-1} \oplus H_{r-1}^\prime +  \Phi_{N_{r}}(t_{r}) \]
where $H_{r-1}^\prime$ is an independent copy of $H_{r-1}$ and $\Phi_{N_{r}}(t_{r}) $ a $ N_{r}\times N_{r} $ Dyson Brownian motion at time $ t_{r} $.

\begin{proof}[Proof of the lower bound]
After slightly adjusting $K_l$ by a small error, the lower bound implied by Lemma~\ref{thm:coarselocallaw} immediately yields that
\[\frac{1}{N_r} \im \tr \left( H_{r-1} \oplus H_{r-1}^\prime - z \right)^{-1} \geq K_l\]
for all $z \in \tilde{\Omega}_r$ with probability $1 - N_r^{-p}$. We have $t_r \eta_r \to \infty$ and $t_r \log (N_r) \to 0$ as $r \to \infty$ by~\eqref{alphatilde}. Since also
\[\operatorname{dist}(W_0, \partial \tilde{W}) > 0\]
the domains $\Omega_r, \tilde{\Omega}_r$ and the time $t_r$ satisfy the hypothesis~\eqref{eq:s0bound} and~\eqref{eq:Tlower} of Theorem~\ref{thm:bndprop} for sufficiently large $r$. After adjusting $K_l$ again, we conclude that
\[\mathbb{P}\left( \inf_{z\in \Omega_{r}}  \im S_{r}(  z )\geq  K_l   \right) \geq 1 - N_{r}^{-p} \]
for any $p>0$ when $r$ is large enough.
\end{proof}

\section{Local Green Function Estimates}\label{sec:localresbounds}
In this section we will prove the following stochastic domination bound for the local resolvents, which provides the main technical tool in the proofs of Theorems~\ref{thm:eigenfunctions} and~\ref{thm:measures}.
\begin{theorem}\label{thm:locresbound} Let $x \in \{1, \dots, N_n\}$. Then
\[\sup_{z \in W + i(\eta_n, 2)} \im G_n(x,  z) \prec 1.\]
\end{theorem}

Since the distribution of $ G_r(x,  z) $ does not depend on $ x $, we restrict our attention to 
\[G_r(z) = G_r(1, z)\]
without loss of generality. The proof again employs the observation~\eqref{eq:dbm}
\[H_{r+1} = H_{r} \oplus H_{r}^\prime +  \Phi_{N_{r+1}}(t_{r+1})\]
and controls
\[G_{r+1}(t, z) = \langle \delta_1, \left(H_{r} \oplus H_{r}^\prime + \Phi_{N_{r+1}}(t) - z\right)^{-1} \delta_1 \rangle\]
using the results of Section~\ref{sec:dbm}.

\begin{proof}[Proof of Theorem~\ref{thm:locresbound}]
Let $m = (1-\delta)n$ with $\delta \in (0,1)$ to be specified later. Suppose that $W$ takes the form $W = [W_1, W_2]$ so that, for large enough $n$, the slightly fattened spectral domain
\[D_r = \left(W_1 - b, W_2 + b\right) + i\left(\eta_r, 2 + b \right)\]
with
\[b = CK_u \log(N_r) \sum_{j=m}^r t_j\]
is still contained in the spectral domain $\Omega_r$ of Theorem~\ref{thm:locallaw} for all $m \le r \le n$. The proof will proceed by recursively controlling the moments
\[X(r, p) = \left\| \sup_{z \in D_r} \im G_r(z)\right\|_p = \left( \ee \left|  \sup_{z \in D_r} \im G_r(z)\right|^p \right)^{1/p}.\]
We begin by restricting the expectation in the defintion of $X(r+1,p)$ to the event $\mathcal{A}_{r+1}$ of~\eqref{eq:aevent} for the Dyson Brownian motion $\Phi_{N_{r+1}}(t)$. Combined with the trival bound on $G_{r+1}$, this yields
\[X(r+1, p) \le \left\| \sup_{z \in D_{r+1}} \im G_{r+1}(z)1_{\mathcal{A}_{r+1}} \right\|_p  + \eta_{r+1}^{-1} \pp\left(\mathcal{A}_{r+1}^c \right)^{1/p}.\]
Let
\[\gtil_{t_{r+1}}(w) = G_{r+1}(t_{r+1} \wedge \tau(w), \xi_{t_{r+1} \wedge \tau(w)}(w)).\]
denote the evolution of $G_t$ along the stopped characteristic featured in Proposition~\ref{thm:gfuncmoment}. On the event $\mathcal{A}_{r+1}$,  Theorem~\ref{thm:bndprop} shows that there exists some $D_r^\prime \subset D_r$ such that $D_{r+1} \subset \xi_{t_r}(D_r^\prime)$ and $\tau_w \geq t_{r+1}$ for $w \in D_r^\prime$. Combining this with Corollary~\ref{thm:stoppingtimes}, which lets us replace the stopping time $\tau(w)$ by $\tau_w$, we obtain
\[\left\|\sup_{z \in D_{r+1}} \im G_{r+1}(z) 1_{\mathcal{A}_{r+1}}\right\|_p \le \left\| \sup_{w \in D_r^\prime} \im \gtil_{t_{r+1}}(w)\right\|_p.\] 
Let $\Lambda_r \subset D_r^\prime$ be a finite grid such that $|\Lambda_r| \le \eta_{r}^{-8}$ and $\operatorname{dist}(z, \Lambda_r) \le C\eta_{r+1}^{4}$ for all $z \in D_r^\prime$. For $w \in D_r^\prime$, Proposition~\ref{prop:char} shows that $w \to \xi_{t_{r+1}}(w) = \gamma(t_{r+1}, w)$ has Lipschitz constant $C \eta_{r+1}^{-1}$ so the function $w \to  \im \gtil_{t_{r+1}}(w)$ has Lipschitz constant $C\eta_{r+1}^{-3}$.  Our choice of $\Lambda_r$ therefore guarantees that
\beq{eq:firstxineq} X(r+1, p) \le  \left\| \max_{w \in \Lambda_r} \im \gtil_{t_{r+1}}(w)\right\|_p +  C\eta_{r+1} + \eta_{r+1}^{-1} \pp\left(\mathcal{A}_{r+1}^c \right)^{1/p}.
\eeq

We now take some $q> p$, apply Jensen's inequality to the conditional expectation with respect to $H_r, H_r^\prime$, and replace the maximum by a sum:
\begin{align*}\left\| \max_{w \in \Lambda_r} \im \gtil_{t_{r+1}}(w)\right\|_p &\le \ee \left[\conee{\max_{w \in \Lambda_r} \left|\im \gtil_t(w) \right|^q}{H_r, H_r^\prime}^{p/q}\right]^{1/p}\\
&\le \ee \left[  \left( \sum_{w \in \Lambda_r} \conee{ \left|\im \gtil_t(w) \right|^q}{H_r, H_r^\prime}\right)^{p/q}\right]^{1/p}.
\end{align*}
Applying Proposition~\ref{thm:gfuncmoment} to the conditional expectation, we obtain 
\begin{align} \label{eq:secondxineq} \left\| \max_{w \in \Lambda_r} \im \gtil_{t_{r+1}}(w)\right\|_p &\le C(q, r) \ee \left[  \left( \sum_{w \in \Lambda_r} \left|\im G_r(w) \right|^q\right)^{p/q}\right]^{1/p} \nonumber\\
&\le C(q,r) \left|\Lambda_r\right|^{1/q}   X(r, p),
\end{align}
where
\[C(q, r) = \left(1 - C q \sqrt{\frac{10}{K_l N_r \eta_r}}\right)^{-1}.\]
Putting~\eqref{eq:firstxineq} and~\eqref{eq:secondxineq} together proves that
\[X(r+1,p) \le C(q,r) \left|\Lambda_r\right|^{1/q} X(r,p) + \eta_{r+1}^{-1} \pp\left(\mathcal{A}_{r+1}^c \right)^{1/p} +  C\eta_{r+1}.\]
Provided that $p$ is fixed and $q$ grows only polynomially in $n$ the terms
\[\prod_{r=m}^{n-1} C(q,r) \qquad \mbox{and} \qquad \sum_{r=m}^{n-1} \left( \eta_{r+1}^{-1} \pp\left(\mathcal{A}_{r+1}^c \right)^{1/p} +  C\eta_{r+1}\right)\]
remain bounded uniformly in $n$ by Lemma~\ref{thm:continuum}. Therefore, the previous bound may be iterated to obtain
\begin{align*}X(n, p) &\le  \left|\Lambda_n\right|^{(n-m)/q} \prod_{r=m+1}^{n-1} C(q,r) X(m, p) + C\sum_{r=m+1}^{n-1} \left( \eta_{r+1}^{-1} \pp\left(\mathcal{A}_{r+1}^c \right)^{1/p} +  C\eta_{r+1}\right)\\
&\le C \eta_n^{-8(n-m)/q} (1 + X(m,p)).
\end{align*}

Given $\theta > 0$ and $\tilde{p} < \infty$, we have to show that
\[\pp\left(\sup_{z \in D_n} \im G(z) > N_n^\theta \right) \le N_n^{-\tilde{p}}\]
for sufficiently large $n$. We choose $\delta \in (0,1)$ such that $(1 - \delta) < \theta/4$ and let $q = 32 \, \theta^{-1} (n-m)$. 
Then for any $p > 0$ and $n$ large enough we have
\[X(n, p) \le C \eta_n^{-\theta/4}(1 + X(m, p)) \le C \eta_n^{-\theta/4} \eta_m^{-1}  \le N_n^{\theta/2}.\]
Taking $p = 2\theta^{-1}\tilde{p}$ and using Markov's inequality finishes the proof.
\end{proof}

\section{Proofs of the Main Results}\label{sec:proof}
We will now turn to translating the resolvent bounds in Theorems~\ref{thm:locallaw} and~\ref{thm:locresbound} into the spectral statements in the main results of this paper.

\begin{proof}[Proof of Theorem~\ref{thm:eigenfunctions}] The bounds for the eigenfunctions follow from Theorem~\ref{thm:locresbound} and the inequality
\[|\psi_E(x)|^2 \le \sum_{\lambda \in \sigma(H_n)} \frac{(2\eta_n)^2}{(E-\lambda)^2 + (2\eta_n)^2} |\psi_\lambda(x)|^2 = 2\eta_n \, \im G_n(x, x; E + i2\eta_n).\]
Given $\theta, p > 0$, we choose $\alpha < \theta$ so that
\[\pp\left( \sup_{E \in W} |\psi_E(x)|^2 > N_n^{\theta - 1} \right) \le \pp\left(\sup_{z \in W + i(\eta_n, 2)}  \im G_n(x,x;z) > N_n^{\theta - \alpha} \right) \le N_n^{-p} \]
for large enough $n$. The assertion of Theorem~\ref{thm:eigenfunctions} now follows from the union bound.
\end{proof}

We now turn to the infinite-volume operator~\eqref{eq:defHinfty}
and a proof of Theorem~\ref{thm:measures}. For each $ r $  the operators in~\eqref{eq:defHinfty} are direct sums
\[ \Phi_{\infty,r} = \bigoplus_{k=1}^\infty  \Phi_r^{(k)} \]
of independent $ N_r \times N_r $ Gaussian Orthogonal Ensembles $  \Phi_r^{(k)} $. The expected norm of the first block $ \ee\|\Phi_r^{(1)}\|$ is bounded uniformly in $r \geq 0$  (\cite[Chap.~2.3]{MR2906465}) so that
\[ \sum_{r} t_r^{\beta/2}  \ee\|\Phi_r^{(1)}\| < \infty\]
for any $\beta > 0$. Hence, there is  $C_\beta(\omega) $ with $ \ee C_\beta < \infty $ such that
for all $r $
\beq{eq:normtailbound}\big\|\Phi_r^{(1)}\big\| \le C_\beta(\omega) t_r^{-\beta/2}. \eeq
In a slight abuse of notation, extending $ \Phi_r^{(1)} $ to $ \ell^2(\mathbb{N}) $ in the canonical fashion, the symmetric operator
\[S_1 = \sum_{r=0}^\infty \sqrt{t_r} \Phi_r^{(1)}  \]
is thus almost surely bounded. Since $H - S_1$ is a direct sum of finite-dimensional symmetric matrices, $H$ is almost surely essentially self-adjoint on the finitely supported functions in $\ell^2(\nn)$.

\begin{proof}[Proof of Theorem~\ref{thm:measures}] Since the distribution of the spectral measure $\mu_x$ of $\delta_x$ for
$H$ does not depend on $x \in \nn$, it suffices to prove the continuity of $\mu_1$.  The tail bound~\eqref{eq:normtailbound} with $ \beta \in (0,1) $ implies that the infinite-volume Green function
\[G(z) = \langle \delta_1, (H-z)^{-1} \delta_1 \rangle\]
can be approximated by its finite-volume counterparts:
\begin{align*}|G(z) - G_n(z)| &= \left| \langle \delta_1, (H-z)^{-1}(H-H_n)\left(H_n - z\right)^{-1} \delta_1 \rangle \right|\\
& \le \eta^{-1} \left\| \sum_{r=n+1}^\infty \sqrt{t_r} \, \Phi_r^{(1)} \left(H_n - z\right)^{-1} \delta_1\right\|\\
&\le C_\beta(\omega)\,  \eta^{-2} \sum_{r=n+1}^\infty t_r^{(1-\beta)/2}\\
&\le C^\prime_\beta(\omega) \, \eta^{-2}  t_n^{(1-\beta)/2}.
\end{align*}
On the other hand, combining the Borel-Cantelli lemma with Theorem~\ref{thm:locresbound} shows that for any $\gamma > 0$ we almost surely have
\[ \sup_{z \in W + i(\eta_n, 2)} \im G_n(z)  \le  N_n^\gamma\]
for all sufficiently large $n$.
Given $E \in W$, $\eta \in (0, 1)$, and $\theta \in (0,1)$ we choose
\[n \geq \log_2(\eta^{-1}) \cdot \max\left\{\frac{4}{(1-\beta)(1+\epsilon)}, \frac{1}{(1-\alpha)}\right\},\]
so $\eta_n \le \eta$ and
\[\sup_{E \in W} |G(E+i\eta) - G_n(E+i\eta)| \le C^\prime_\beta(\omega).\]
Choosing $\gamma$ sufficiently small this implies that
\begin{align*}
\mu_{\delta_1}(E-\eta, E+\eta) &\le 2\eta \, \im G(E+i\eta)\\
&\le 2\eta \sup_{z \in D_n} \im G_n(z) + C^\prime_\beta(\omega) \eta \\
&\le (2+C^\prime_\beta(\omega)) \eta \, N_n^\gamma\\
&\le(2 + C^\prime_\beta(\omega)) \eta^{1-\theta}.
\end{align*}
\end{proof}

Finally, Theorem~\ref{thm:localstats} is a direct application of the work of Landon, Sosoe, and Yau~\cite{landonsosoeyauarxiv}.

\begin{proof}[Proof of Theorem~\ref{thm:localstats}]
We choose $0 < \alpha < \epsilon$ so that $\eta_n \ll t_n$ and note that the bounds of Theorem~\ref{thm:locallaw} extend immediately to the empirical eigenvalue measure of $H_{n-1} \oplus H_{n-1}^\prime$. This means that the analysis of~\cite{landonsosoeyauarxiv} applies, up to the minor technical point that the result of~\cite{landonsosoeyauarxiv} is stated for Gaussian perturbations of deterministic matrices. In the present setting this means that~\cite{landonsosoeyauarxiv} applies when the $k$-point function~\eqref{eq:rhokdef} is the marginal of the symmetrized eigenvalue density conditioned on $H_{n-1} \oplus H_{n-1}^\prime$ and that the Stieltjes transform of the scaling density $\rho_{n, fc}(E)$ in~\eqref{eq:psidef} solves
\[M(z) = \frac{1}{N_n} \sum_{\lambda \in \sigma(H_{n-1} \oplus H_{n-1}^\prime)}  \frac{1}{\lambda - z - t_n M(z)}\]
rather than~\eqref{eq:Mdef}. As mentioned in the remarks of~\cite{landonyauarxiv}, this problem is easily remedied since the local law in Lemma~\ref{thm:coarselocallaw} and the arguments of Lemma 3.6 of~\cite{MR3502606} show that the scaling factor $\rho_{n, fc}(E)$ is typically close to its counterpart in the statement of Theorem~\ref{thm:localstats}. This observation and simple weak convergence arguments then imply the result as stated.
\end{proof}

\minisec{Acknowledgment}
We would like to thank an anonymous referee for helping us clarify the proof of Theorem~\ref{thm:bndprop}. This work was supported by the DFG (WA 1699/2-1).

\bibliographystyle{abbrv}
\bibliography{References}

\begin{thebibliography}{10}

\bibitem{MR2257129}
M.~Aizenman, R.~Sims, and S.~Warzel.
\newblock Stability of the absolutely continuous spectrum of random
  {S}chr\"odinger operators on tree graphs.
\newblock {\em Probab. Theory Related Fields}, 136(3):363--394, 2006.

\bibitem{MR3055759}
M.~Aizenman and S.~Warzel.
\newblock Resonant delocalization for random {S}chr\"{o}dinger operators on
  tree graphs.
\newblock {\em J. Eur. Math. Soc. (JEMS)}, 15(4):1167--1222, 2013.

\bibitem{MR3364516}
M.~Aizenman and S.~Warzel.
\newblock {\em Random operators: Disorder effects on quantum spectra and
  dynamics}, volume 168 of {\em Graduate Studies in Mathematics}.
\newblock American Mathematical Society, Providence, RI, 2015.

\bibitem{bauerschmidt2}
R.~Bauerschmidt, J.~Huang, A.~Knowles, and H.-T. Yau.
\newblock Bulk eigenvalue statistics for random regular graphs.
\newblock {\em Ann. Probab.}, 45(6A):3626--3663, 2017.

\bibitem{bauerschmidt1}
R.~Bauerschmidt, A.~Knowles, and H.-T. Yau.
\newblock Local semicircle law for random regular graphs.
\newblock {\em Comm. Pure Appl. Math.}, 70(10):1898--1960, 2017.

\bibitem{benignibourgade}
L.~Benigni.
\newblock Eigenvectors distribution and quantum unique ergodicity for deformed
  {W}igner matrices.
\newblock Preprint available at ar{X}iv:1711.07103, 2017.

\bibitem{MR1552611}
P.~M. Bleher and Y.~G. Sinai.
\newblock Critical indices for {D}yson's asymptotically-hierarchical models.
\newblock {\em Comm. Math. Phys.}, 45(3):247--278, 1975.

\bibitem{PhysRevE.98.042116}
E.~Bogomolny and M.~Sieber.
\newblock Power-law random banded matrices and ultrametric matrices:
  Eigenvector distribution in the intermediate regime.
\newblock {\em Phys. Rev. E}, 98:042116, Oct 2018.

\bibitem{MR3695802}
P.~Bourgade, L.~Erd\H{o}s, H.-T. Yau, and J.~Yin.
\newblock Universality for a class of random band matrices.
\newblock {\em Adv. Theor. Math. Phys.}, 21(3):739--800, 2017.

\bibitem{bandmatrix2}
P.~Bourgade, F.~Yang, H.-T. Yau, and J.~Yin.
\newblock Random band matrices in the delocalized phase, {II}: {G}eneralized
  resolvent estimates.
\newblock Preprint available at ar{X}iv:1807.01562, 2018.

\bibitem{bandmatrix1}
P.~Bourgade, H.-T. Yau, and J.~Yin.
\newblock Random band matrices in the delocalized phase, {I}: {Q}uantum unique
  ergodicity and universality.
\newblock Preprint available at ar{X}iv:1807.01559, 2018.

\bibitem{MR1063180}
A.~Bovier.
\newblock The density of states in the {A}nderson model at weak disorder: a
  renormalization group analysis of the hierarchical model.
\newblock {\em J. Statist. Phys.}, 59(3-4):745--779, 1990.

\bibitem{MR1143413}
D.~Brydges, S.~N. Evans, and J.~Z. Imbrie.
\newblock Self-avoiding walk on a hierarchical lattice in four dimensions.
\newblock {\em Ann. Probab.}, 20(1):82--124, 1992.

\bibitem{PhysRevLett.64.1851}
G.~Casati, L.~Molinari, and F.~Izrailev.
\newblock Scaling properties of band random matrices.
\newblock {\em Phys. Rev. Lett.}, 64:1851--1854, Apr 1990.

\bibitem{MR3665217}
M.~Disertori and M.~Lager.
\newblock Density of states for random band matrices in two dimensions.
\newblock {\em Ann. Henri Poincar\'e}, 18(7):2367--2413, 2017.

\bibitem{MR0148397}
F.~J. Dyson.
\newblock A {B}rownian-motion model for the eigenvalues of a random matrix.
\newblock {\em J. Mathematical Phys.}, 3:1191--1198, 1962.

\bibitem{MR0436850}
F.~J. Dyson.
\newblock Existence of a phase-transition in a one-dimensional {I}sing
  ferromagnet.
\newblock {\em Comm. Math. Phys.}, 12(2):91--107, 1969.

\bibitem{MR3085669}
L.~Erd\H{o}s, A.~Knowles, H.-T. Yau, and J.~Yin.
\newblock Delocalization and diffusion profile for random band matrices.
\newblock {\em Comm. Math. Phys.}, 323(1):367--416, 2013.

\bibitem{MR3068390}
L.~Erd\H{o}s, A.~Knowles, H.-T. Yau, and J.~Yin.
\newblock The local semicircle law for a general class of random matrices.
\newblock {\em Electron. J. Probab.}, 18:no. 59, 58, 2013.

\bibitem{MR2810797}
L.~Erd\H{o}s, B.~Schlein, and H.-T. Yau.
\newblock Universality of random matrices and local relaxation flow.
\newblock {\em Invent. Math.}, 185(1):75--119, 2011.

\bibitem{MR2274470}
R.~Froese, D.~Hasler, and W.~Spitzer.
\newblock Absolutely continuous spectrum for the {A}nderson model on a tree: a
  geometric proof of {K}lein's theorem.
\newblock {\em Comm. Math. Phys.}, 269(1):239--257, 2007.

\bibitem{PhysRevLett.67.2405}
Y.~V. Fyodorov and A.~D. Mirlin.
\newblock Scaling properties of localization in random band matrices: A
  \ensuremath{\sigma}-model approach.
\newblock {\em Phys. Rev. Lett.}, 67:2405--2409, Oct 1991.

\bibitem{1742-5468-2009-12-L12001}
Y.~V. Fyodorov, A.~Ossipov, and A.~Rodriguez.
\newblock The {A}nderson localization transition and eigenfunction
  multifractality in an ensemble of ultrametric random matrices.
\newblock {\em Journal of Statistical Mechanics: Theory and Experiment},
  2009(12):L12001, 2009.

\bibitem{MR649813}
K.~Gaw\c{e}dzki and A.~Kupiainen.
\newblock Renormalization group study of a critical lattice model. {II}. {T}he
  correlation functions.
\newblock {\em Comm. Math. Phys.}, 83(4):469--492, 1982.

\bibitem{MR1781846}
A.~Guionnet and O.~Zeitouni.
\newblock Concentration of the spectral measure for large matrices.
\newblock {\em Electron. Comm. Probab.}, 5:119--136, 2000.

\bibitem{hemarcozzi}
Y.~He and M.~Marcozzi.
\newblock Diffusion profile for random band matrices: a short proof.
\newblock Preprint available at ar{X}iv:1804.09446, 2018.

\bibitem{1064-5616-206-1-93}
W.~Kirsch and L.~A. Pastur.
\newblock On the analogues of {S}zeg{\H{o}'}s theorem for ergodic operators.
\newblock {\em Sbornik: Mathematics}, 206(1):93, 2015.

\bibitem{MR1302384}
A.~Klein.
\newblock Absolutely continuous spectrum in the {A}nderson model on the {B}ethe
  lattice.
\newblock {\em Math. Res. Lett.}, 1(4):399--407, 1994.

\bibitem{MR2864550}
A.~Klein and C.~Sadel.
\newblock Absolutely continuous spectrum for random {S}chr\"{o}dinger operators
  on the {B}ethe strip.
\newblock {\em Math. Nachr.}, 285(1):5--26, 2012.

\bibitem{MR2352276}
E.~Kritchevski.
\newblock Hierarchical {A}nderson model.
\newblock In {\em Probability and mathematical physics}, volume~42 of {\em CRM
  Proc. Lecture Notes}, pages 309--322. Amer. Math. Soc., Providence, RI, 2007.

\bibitem{landonsosoeyauarxiv}
B.~Landon, P.~Sosoe, and H.-T. Yau.
\newblock Fixed energy universality for {D}yson {B}rownian motion.
\newblock Preprint available at arXiv:1609.09011, 2016.

\bibitem{landonyauarxiv}
B.~Landon and H.-T. Yau.
\newblock Convergence of {L}ocal {S}tatistics of {D}yson {B}rownian {M}otion.
\newblock {\em Comm. Math. Phys.}, 355(3):949--1000, 2017.

\bibitem{MR1423040}
Y.~Last.
\newblock Quantum dynamics and decompositions of singular continuous spectra.
\newblock {\em J. Funct. Anal.}, 142(2):406--445, 1996.

\bibitem{MR3502606}
J.~O. Lee, K.~Schnelli, B.~Stetler, and H.-T. Yau.
\newblock Bulk universality for deformed {W}igner matrices.
\newblock {\em Ann. Probab.}, 44(3):2349--2425, 2016.

\bibitem{PhysRevE.54.3221}
A.~D. Mirlin, Y.~V. Fyodorov, F.-M. Dittes, J.~Quezada, and T.~H. Seligman.
\newblock Transition from localized to extended eigenstates in the ensemble of
  power-law random banded matrices.
\newblock {\em Phys. Rev. E}, 54:3221--3230, Oct 1996.

\bibitem{MR1463464}
S.~Molchanov.
\newblock Hierarchical random matrices and operators. {A}pplication to
  {A}nderson model.
\newblock In {\em Multidimensional statistical analysis and theory of random
  matrices ({B}owling {G}reen, {OH}, 1996)}, pages 179--194. VSP, Utrecht,
  1996.

\bibitem{doi:10.1093/imrn/rnx145}
R.~Peled, J.~Schenker, M.~Shamis, and S.~Sodin.
\newblock On the {W}egner orbital model.
\newblock {\em International Mathematics Research Notices}, page rnx145, 2017.

\bibitem{MR1725357}
D.~Revuz and M.~Yor.
\newblock {\em Continuous martingales and {B}rownian motion}, volume 293 of
  {\em Grundlehren der Mathematischen Wissenschaften [Fundamental Principles of
  Mathematical Sciences]}.
\newblock Springer-Verlag, Berlin, third edition, 1999.

\bibitem{MR3510466}
C.~Sadel.
\newblock Anderson transition at two-dimensional growth rate on antitrees and
  spectral theory for operators with one propagating channel.
\newblock {\em Ann. Henri Poincar\'{e}}, 17(7):1631--1675, 2016.

\bibitem{MR2525652}
J.~Schenker.
\newblock Eigenvector localization for random band matrices with power law band
  width.
\newblock {\em Comm. Math. Phys.}, 290(3):1065--1097, 2009.

\bibitem{MR3824956}
M.~Shcherbina and T.~Shcherbina.
\newblock Universality for 1d {R}andom {B}and {M}atrices: {S}igma-{M}odel
  {A}pproximation.
\newblock {\em J. Stat. Phys.}, 172(2):627--664, 2018.

\bibitem{MR2726110}
S.~Sodin.
\newblock The spectral edge of some random band matrices.
\newblock {\em Ann. of Math. (2)}, 172(3):2223--2251, 2010.

\bibitem{MR2906465}
T.~Tao.
\newblock {\em Topics in random matrix theory}, volume 132 of {\em Graduate
  Studies in Mathematics}.
\newblock American Mathematical Society, Providence, RI, 2012.

\bibitem{MR3649447}
P.~von Soosten and S.~Warzel.
\newblock Renormalization {G}roup {A}nalysis of the {H}ierarchical {A}nderson
  {M}odel.
\newblock {\em Ann. Henri Poincar\'e}, 18(6):1919--1947, 2017.

\bibitem{resflow}
P.~von Soosten and S.~Warzel.
\newblock The phase transition in the ultrametric ensemble and local stability
  of {D}yson {B}rownian motion.
\newblock {\em Electron. J. Probab.}, 23:1--24, 2018.

\bibitem{proceedings}
P.~von Soosten and S.~Warzel.
\newblock Singular spectrum and recent results on hierarchical operators.
\newblock In {\em Mathematical {P}roblems in {Q}uantum {P}hysics}, volume 717
  of {\em Contemp. Math.}, pages 215--225. Amer. Math. Soc., Providence, RI,
  2018.

\bibitem{nonergodic}
P.~von Soosten and S.~Warzel.
\newblock Non-ergodic delocalization in the {R}osenzweig-{P}orter model.
\newblock {\em Lett. Math. Phys.}, 109(4):905--922, 2019.

\bibitem{MR639135}
F.~Wegner.
\newblock Bounds on the density of states in disordered systems.
\newblock {\em Z. Phys. B}, 44(1-2):9--15, 1981.

\bibitem{bandmatrix3}
F.~Yang and J.~Yin.
\newblock Random band matrices in the delocalized phase, {III}: {A}veraging
  fluctuations.
\newblock Preprint available at ar{X}iv:1807.02447, 2018.

\end{thebibliography}

\newpage
\noindent Per von Soosten\\
Munich Center for Quantum Science and Technology, and\\
Zentrum Mathematik, TU M\"{u}nchen\\
Boltzmannstra{\ss}e 3, 85747 Garching, Germany\\
\verb+vonsoost@ma.tum.de+ 
\bigskip

\noindent Simone Warzel\\
Munich Center for Quantum Science and Technology, and\\
Zentrum Mathematik, TU M\"{u}nchen\\
Boltzmannstra{\ss}e 3, 85747 Garching, Germany\\
\verb+warzel@ma.tum.de+
\end{document}